\documentclass[12pt]{article}

\usepackage{geometry}
\usepackage{color}
\usepackage{graphicx}
\usepackage{amsmath}
\usepackage{amsfonts}
\usepackage{amsthm}
\usepackage{amscd}
\usepackage{epsfig}
\usepackage{amssymb}
\usepackage{slashed}
\usepackage{tabularx}
\usepackage{calligra}
\usepackage{enumerate}
\usepackage{enumitem}
\usepackage{hyperref}
\usepackage{float}
\usepackage{stackrel}
\usepackage[T1]{fontenc}
\usepackage{longtable}
\usepackage{amsthm}
\usepackage{mathrsfs}
\usepackage{verbatim} 
\usepackage{fancyhdr}
\usepackage{tikz}
\usepackage{tikz-cd}
\usetikzlibrary{matrix}
\usetikzlibrary{positioning}
\usepackage[all]{xy}
\usepackage{units}
\usepackage{textcomp}
\usepackage{turnstile}
\usepackage{ stmaryrd }
\usepackage{tikz-3dplot}

\theoremstyle{plain}
\newtheorem{theorem}{Theorem}
\newtheorem{proposition}[theorem]{Proposition}

\newtheorem{definition}[theorem]{Definition}
\newtheorem{remark}[theorem]{Remark}

\numberwithin{equation}{section}

\def\d{{\rm d}}
\def\i{{\rm i}}
\def\CP{\mathbb{CP}}

\newcommand{\C}{\mathscr{C}}
\newcommand{\M}{\mathcal{M}}

\newcommand{\PT}{\mathcal{PT}}
\newcommand{\K}{\mathcal{K}}

\begin{document}

\title{Teukolsky equations, twistor functions, and conformally self-dual spaces}
\author{Bernardo Araneda\footnote{Email: \texttt{bernardo.araneda@aei.mpg.de}}  \\
Max-Planck-Institut f\"ur Gravitationsphysik \\ 
(Albert-Einstein-Institut), Am M\"uhlenberg 1, \\
D-14476 Potsdam, Germany}

\date{\today}

\maketitle

\begin{abstract}
We prove a correspondence, for Riemannian manifolds with self-dual Weyl tensor, between twistor functions and solutions to the Teukolsky equations for any conformal and spin weights. In particular, we give a contour integral formula for solutions to the Teukolsky equations, and we find a recursion operator that generates an infinite family of solutions and leads to the construction of \v{C}ech representatives and sheaf cohomology classes on twistor space. 
Apart from the general conformally self-dual case, examples include self-dual black holes, scalar-flat K\"ahler surfaces, and quaternionic-K\"ahler metrics, where we map the Teukolsky equation to the conformal wave equation, establish new relations to the linearised Przanowski equation, and find new classes of quaternionic deformations.
\end{abstract}

\section{Introduction}

The Teukolsky equations \cite{Teukolsky} are scalar partial differential equations describing perturbations of algebraically special space-times, notably rotating black holes. They are instrumental to modern studies in mathematical relativity and gravitational physics, including black hole stability, classical and quantum fields on black holes, propagation of gravitational waves, etc.
In Newman-Penrose notation \cite{PR1}, the Teukolsky equation for the extreme spin-weight component, $\chi$, of a spin-$|s|$ field (with $s\in\frac{1}{2}\mathbb{Z}$), is 
\begin{equation}\label{TeukEq}
\begin{aligned}
{\rm T}_{s}\chi :={}
& 2\left[ \left( D-(2s-1)\varepsilon+\tilde\varepsilon-2s\rho-\tilde\rho \right) \left( D'+2s\varepsilon'-\rho'\right) \right. \\
& \left. - (\delta-(2s-1)\beta+\tilde\beta'-2s\tau-\tilde\tau ) \left( \delta'+2s\beta'-\tau' \right) 
- 2(s-\tfrac{1}{2})(s-1)\Psi_2 \right]\chi = 0
\end{aligned}
\end{equation}
where $D,D',\delta,\delta'$ are derivatives along a principal null tetrad, $\Psi_2$ is the middle component of the Weyl tensor with respect to this tetrad, and the rest of the symbols in \eqref{TeukEq} are the spin coefficients of the tetrad. For $s=0$, \eqref{TeukEq} reduces to the covariant scalar wave equation.

More generally, the Teukolsky equation \eqref{TeukEq} is part of a larger family of scalar  equations associated to space-time perturbations. Other examples include: the spin-weight zero component of an electromagnetic perturbation satisfies the Fackerell-Ipser equation \cite{Fackerell}; gravitational perturbations of static black holes are described by the Regge-Wheeler and Zerilli equations \cite{ReggeWheeler, Zerilli}; and deformations of hyper-K\"ahler and quaternionic-K\"ahler 4-metrics are encoded, respectively, in the linearised heavenly and Przanowski equations \cite{DM, Hoegner, Alexandrov2}. All of the above (including \eqref{TeukEq}) are particular cases of a generalised ``wave'' equation,
\begin{align}\label{BoxCIntro}
 (\Box^{\C}+c\Psi_2)\chi = 0
\end{align}
where $c$ is some constant\footnote{The only case in which $c$ is not a constant is the Zerilli equation, where $c$ is an integro-differential operator. This is not relevant for the present paper, because of the assumption \eqref{SDWeyl}.} and $\Box^{\C}=g^{ab}\C_a\C_b$ is the ``wave'' operator associated to a conformally and GHP invariant connection $\C_a$ \cite{A18, A21}. The field $\chi$ is a weighted scalar field that carries both conformal weight $w\in\frac{\mathbb{Z}}{2}$ and spin weight $s=\frac{p}{2}\in\frac{\mathbb{Z}}{2}$. For simplicity, we will refer to \eqref{BoxCIntro} as the ``Teukolsky equation for weights $(w,p)$'' even though it includes not only \eqref{TeukEq} but also all of the other cases mentioned above (simply by adjusting $w$ and $p$).

Despite a number of remarkable properties of the Teukolsky equation, the construction of general solutions is a difficult problem. In this work, we are interested in the study of this equation on a special class of geometries: those with self-dual Weyl tensor,
\begin{align}\label{SDWeyl}
 *C_{abcd} = C_{abcd},
\end{align}
also called conformally self-dual. Even though the assumption \eqref{SDWeyl} excludes the standard Lorentzian black holes, the Teukolsky equation on a geometry satisfying \eqref{SDWeyl} is non-trivial in that it does not reduce to the ordinary scalar wave equation. In terms of \eqref{TeukEq}-\eqref{BoxCIntro}, the condition \eqref{SDWeyl} implies $\Psi_2=0$ (but no restrictions on the Ricci tensor). 

Our interest in manifolds satisfying \eqref{SDWeyl} is that they have a twistor space. In a number of situations, twistor methods have led to remarkable constructions of solutions to (linear and non-linear) differential equations on space-time in terms of {\em free} holomorphic twistor functions, see e.g. \cite{Atiyah}. A prominent example at the linear level is the Penrose transform: an isomorphism between the space of massless fields on Minkowski and certain twistor sheaf cohomology groups \cite{Eastwood}. In this spirit, our main result in connection to the Teukolsky equation is the following:

\begin{theorem}\label{theorem:main}
Let $(\M,g_{ab})$ be conformally self-dual \eqref{SDWeyl}, let $\mathcal{PT}$ be its twistor space, and let $w\in\frac{\mathbb{Z}}{2}$, $p\in\mathbb{Z}$. Fix an integrable complex structure $J^{a}{}_{b}$ on $\M$, with principal spinors $o_A,\iota_A$. 
\begin{enumerate}[noitemsep, nolistsep]
\item Let $\Psi$ be a representative of the \v{C}ech cohomology group $\check{H}^{1}(\PT,\mathcal{O}(2w))$, and let 
\begin{align}\label{CountIntIntro}
 \chi(x^a):= \frac{1}{2\pi\i}\oint_{\Gamma}\frac{\Psi(x^a,\lambda)\lambda_{A}\d\lambda^{A}}{ (o_{B}\lambda^{B})^{w+1-\frac{p}{2}} (\iota_{B}\lambda^{B})^{w+1+\frac{p}{2}}},
\end{align}
where $x^a\in\M$, $\Gamma\subset\CP^1$. Then $\chi$ solves the Teukolsky equation for weights $(w,p)$. 
\label{item:solution}
\item Let $\chi$ be a solution to the Teukolsky equation for weights $(w,p)$. Then $\chi$ generates a \v{C}ech representative $\Psi$ in $\check{H}^{1}(\PT,\mathcal{O}(2w))$, such that formula \eqref{CountIntIntro} holds.
\item The space-time field $\chi$ of the previous item is invariant under changes of \v{C}ech representative $\Psi$ if and only if $w\pm\frac{p}{2}\in-\mathbb{N}$.
\label{item:cohomology}
\end{enumerate}
\end{theorem}
 
In the statement of the above result, we are using the fact that \v{C}ech representatives are local holomorphic sections of $\mathcal{O}(2w)$ over twistor space. The correspondence with cohomology classes is more subtle as it implies a particular relation between conformal and spin weights, as described by item \ref{item:cohomology}. We also see that the conformal weight characterising a Teukolsky equation plays an important role in that it determines the Chern class of the line bundle involved in the cohomology group $\check{H}^{1}(\PT,\mathcal{O}(2w))$. Finally, we note that even though a space satisfying \eqref{SDWeyl} admits many (local) integrable complex structures, the construction in theorem \ref{theorem:main} depends on a choice of a single one, which we see as a desirable feature since non-self-dual black holes (which motivate this study) have only one such structure (for a fixed orientation).

Theorem \ref{theorem:main} will be established following a number of intermediate results. After some preliminary material in section \ref{Sec:preliminaries}, in section \ref{Sec:TwistorTheory} we will first see that, for generic Hermitian manifolds, there is a close relationship between the {\em twistor distribution} on the spin bundle and the connection $\C_{a}$ on space-time (prop. \ref{prop:Hermitian}). Then under the assumption \eqref{SDWeyl}, the twistor distribution has a non-trivial kernel (twistor functions) and this will lead to formula \eqref{CountIntIntro} via the relation to $\C_a$ (prop. \ref{prop:TF1}). We will then find a symmetry operator that generates an infinite family of solutions to \eqref{BoxCIntro} for different weights (prop. \ref{prop:RO}), and this will in turn lead to a recursive construction of twistor functions (prop. \ref{prop:TF2} and \ref{prop:ROTF}), following pioneering methods in the twistor literature \cite{MW, DM, Hoegner}. 
Having established the above, a natural question is whether the correspondence between twistor functions (thought of as \v{C}ech representatives) and solutions to \eqref{BoxCIntro} is one-to-one, this will lead naturally to sheaf cohomology (prop. \ref{prop:cohomology}). In the special case where the geometry is also conformally K\"ahler, the recursion operator will allow us to transform the Teukolsky equation into the conformal wave equation (prop. \ref{prop:CKcase}).

Finally, one of the main applications of the Teukolsky equations is to {\em gravitational} perturbations of {\em vacuum} spaces. Given our assumption \eqref{SDWeyl}, the additional requirement of Ricci-flatness would give a hyper-K\"ahler metric. In that case, a constant spin frame could be chosen and \eqref{TeukEq} would reduce to the ordinary wave equation; to avoid this, one can introduce a cosmological constant, so that there are no parallel spinors and \eqref{TeukEq} is non-trivial\footnote{One can also study the Teukolsky equation on a hyper-K\"ahler metric by choosing a non-constant spin frame, e.g. using a complex structure that is {\em conformally} K\"ahler but {\em not K\"ahler}. Any Gibbons-Hawking space admits such a structure (although other hyper-K\"ahler manifolds such as K3 or Atiyah-Hitchin do not).}. Such a space is called a quaternionic-K\"ahler 4-manifold, and in section \ref{Sec:QK} we will apply our construction to study gravitational perturbations of these spaces, via an operator identity (prop. \ref{prop:SEOT}). We will find new classes of deformations that preserve the self-dual Einstein condition (prop. \ref{prop:metricpert} and \ref{prop:qkdeformations}), and we will also give contour integral formulae for them.

\section{Preliminaries}
\label{Sec:preliminaries}

\subsection{Complex structures and connections}
\label{Sec:connections}

Let $(\M,g_{ab})$ be a four-dimensional, orientable Riemannian manifold, with Levi-Civita connection $\nabla_{a}$. Let $J^{a}{}_{b}$ be a compatible almost-complex structure. The 2-form $g_{ac}J^{c}{}_{b}$ is necessarily self-dual or anti-self-dual. Choosing the latter for concreteness, $J^{a}{}_{b}$ can be written in terms of its principal spinors as 
\begin{align}\label{J}
 J^{a}{}_{b} = \i(o^{A}\iota_{B}+\iota^{A}o_{B})\delta^{A'}_{B'},
\end{align}
where we have chosen the normalisation $o_{A}\iota^{A}=1$. Notice that $o_A,\iota_A$ are only defined up to scale: we have the ``GHP'' freedom
\begin{align}
 o_{A} \to z o_{A}, \qquad \iota_{A} \to z^{-1}\iota_{A}, \qquad z\in\mathbb{C}^{*}.
 \label{GHP}
\end{align}
In Euclidean signature we have $\iota_{A}=o^{\dagger}_{A}$, and $z \in U(1)$.

The $+\i$ and $-\i$ eigenspaces of $J^{a}{}_{b}$ are respectively $T^{+}={\rm span}\{\iota^{A}\partial_{AA'}\}$ and $T^{-}={\rm span}\{o^{A}\partial_{AA'}\}$. The almost-complex structure is integrable iff these eigenspaces are involutive, which can be shown to be equivalent to the ``shear-free'' equations
\begin{align}\label{SFR}
 \iota^{A}\iota^{B}\nabla_{AA'}\iota_{B} = 0, \qquad o^{A}o^{B}\nabla_{AA'}o_{B}=0.
\end{align}
In this case $T^{+}$ and $T^{-}$ are tangent to (complex) surfaces, which are called twistor surfaces. Conversely, given two linearly independent spinor fields satisfying \eqref{SFR} and $o_A\iota^A=1$, the tensor field \eqref{J} is an integrable almost-complex structure.

The triple $(\M,g_{ab},J^{a}{}_{b})$ is called an almost-Hermitian manifold, and Hermitian if $J^{a}{}_{b}$ is integrable. The tensor field $J^{a}{}_{b}$ is conformally invariant: if $g_{ab} \to \hat{g}_{ab}=\Omega^2 g_{ab}$ for some positive scalar field $\Omega$, then $J^{a}{}_{b}$ is also a compatible almost-complex structure for $\hat{g}_{ab}$. The principal spinors transform as 
\begin{align}
 \hat{o}_{A} = \Omega^{1/2}o_{A}, \quad \hat\iota_{A} = \Omega^{1/2}\iota_{A}, 
 \qquad
  \hat{o}^{A} = \Omega^{-1/2}o^{A}, \quad \hat\iota^{A} = \Omega^{-1/2}\iota^{A}.
 \label{CT}
\end{align}
We note that, in a complexified picture, one can choose $\hat{o}^{A}=\Omega^{w^0}o^{A}$ and $\hat\iota^{A}=\Omega^{w^1}\iota^{A}$ where $w^0,w^1$ are any real numbers satisfying $w^0+w^1+1=0$ (see \cite[Section 5.6]{PR1}). The choice $w^0=w^1=-1/2$ in \eqref{CT} ensures that conformal transformations commute with complex conjugation in Euclidean signature.

We define a weighted spinor field, with weights $(w,p)$ where $w$ is conformal weight and $p$ is GHP weight, to be a spinor field $\varphi^{\mathcal{I}}$ ($\mathcal{I}$ represents an arbitrary collection of abstract tensor/spinor indices) which transforms under \eqref{GHP} and \eqref{CT} as 
\begin{align}
 \varphi^{\mathcal{I}} \to \Omega^{w}z^{p} \varphi^{\mathcal{I}}.
 \label{transf}
\end{align}
For example: $o_{A}$ and $\iota_{A}$ have weights $(\frac{1}{2},1)$ and $(\tfrac{1}{2},-1)$ respectively, and $o^{A}$ and $\iota^{A}$ have weights $(-\frac{1}{2},1)$ and $(-\tfrac{1}{2},-1)$. A weighted field is more properly thought of as a section of a complex vector bundle $E_{(w,p)}$ over $\M$. A connection $\C_a$ in this bundle was constructed in \cite{A18, A21}. If $\psi,\varphi_{A}, v_{b}$ are respectively a scalar, spinor, and vector field, all of them with weights $(w,p)$, then $\C_{a}$ is defined by
\begin{subequations}\label{DefC}
\begin{align}
\C_{a}\psi &= \nabla_{a}\psi + (w f_{a} + p P_{a})\psi, \label{Cscalar} \\
\C_{a}\varphi_{B} &= \nabla_{a}\varphi_{B} + (w f_{a} + p P_{a})\varphi_{B} - f_{A'B}\varphi_{A}, \label{Cspinor} \\
\C_{a}v_{b} &= \nabla_{a}v_{b} + (w f_{a} + p P_{a})v_{b} - (f_{a}\delta^{c}_{b}+f_{b}\delta^{c}_{a}-g_{ab}f^{c})v_{c}, \label{Cvector}
\end{align}
\end{subequations}
where $f_{a}$ is the Lee form and $P_{a}$ is a conformally invariant version of the GHP connection 1-form $\omega_{a}$:
\begin{align}\label{ConnectionForms} 
f_{a} := -\tfrac{1}{2}J^{b}{}_{c}\nabla_{b}J^{c}{}_{a}, \qquad 
 P_{a} := \omega_{a} - \tfrac{1}{2}\i J^{b}{}_{a}f_{b}, \qquad
 \omega_{a} := \iota_{B}\nabla_{a}o^{B}.
\end{align}
The action of $\C_{a}$ on fields with a different index structure is deduced from \eqref{DefC} by standard spinor/tensor rules. 

\begin{remark}\label{Rem:propertiesofC}
We collect some identities that will be useful later, cf. \cite{A21} for more details. Let $J^{a}{}_{b}$ be an almost-complex structure, with principal spinors $o_A,\iota_A$.
\begin{enumerate}[noitemsep, nolistsep]
\item The following identities hold (see e.g. \cite[Eq. (2.53)]{A18}):
\begin{subequations}\label{Cdyad}
\begin{align}
 \C_{AA'}o_{B} &= \sigma^{0}_{A'}\iota_{A}\iota_{B}, 
 \qquad \sigma^{0}_{A'}:=o^{A}o^{B}\nabla_{AA'}o_{B}, \\
 \C_{AA'}\iota_{B} &= \sigma^{1}_{A'}o_{A}o_{B}, 
 \qquad \sigma^{1}_{A'}:=\iota^{A}\iota^{B}\nabla_{AA'}\iota_{B}.
\end{align}
\end{subequations}
In particular, from \eqref{SFR}  we see that $J^{a}{}_{b}$ is integrable iff $\sigma^{0}_{A'}=\sigma^{1}_{A'}=0$. 
\label{item-integrableJ}
\item Define the following Dolbeault-like operators
\begin{align}\label{Dolbeault}
 \tilde{\C}_{A'}:=o^{A}\C_{AA'}, 
 \qquad \C_{A'}:=\iota^{A}\C_{AA'}.
\end{align}
If $J^{a}{}_{b}$ is integrable and $\Psi_{ABCD}$ is algebraically special along $o^A,\iota^A$, then 
\begin{align}\label{DolbeaultSquared}
 \tilde{\C}^{A'}\tilde{\C}_{A'}=0, \qquad
 \C^{A'}\C_{A'}=0
\end{align}
when acting on any weighted spinor field. See \cite[Proposition 2.2]{A21}.
\item Suppose that $J^{a}{}_{b}$ is integrable and $\Psi_{ABCD}$ is algebraically special along $o^A,\iota^A$. Let $A_{A'},B_{A'}$ be weighted spinor fields such that 
\begin{align}
 \tilde{\C}^{A'}A_{A'} = 0, \qquad \C^{A'}B_{A'} = 0.
\end{align}
Then there exist, locally, scalar fields $a,b$ such that (cf. \cite[Lemma 2.4]{A21})
\begin{align}\label{potentials}
  A_{A'} = \tilde{\C}_{A'}a, \qquad B_{A'} =  \C_{A'}b.
\end{align}
\label{item:potentials}
\end{enumerate}
\end{remark}

\subsection{Teukolsky operators}
\label{sec:TeukolskyOp}

The connection $\C_{a}$ defines a ``wave'' operator on generic weighted spinor fields:
\begin{align}\label{BoxC}
 \Box^{\C}:=g^{ab}\C_{a}\C_{b}.
\end{align}
\begin{proposition}\label{prop:identitiesC}
Let $(\M,g_{ab},J^{a}{}_{b})$ be an almost-Hermitian manifold. 
\smallskip
\begin{enumerate}[noitemsep, nolistsep]
\item The operator \eqref{BoxC} acting on arbitrary weighted scalar fields with weights $(w,p)$ has the following Newman-Penrose expression:
\begin{equation}\label{BoxCNP}
\begin{aligned}
 \tfrac{1}{2}\Box^{\C} ={}& [D-(p-1)\varepsilon+\tilde\varepsilon+(w-\tfrac{p}{2}+1)\rho-\tilde\rho][D'+p\varepsilon+(w+\tfrac{p}{2})\rho'] \\
 & - [\delta-(p-1)\beta+\tilde\beta'+(w-\tfrac{p}{2}+1)\tau-\tilde\tau'][\delta'+p\beta'+(w+\tfrac{p}{2})\tau'] \\
 & + \tfrac{3}{2}p\Psi_2 - p(\sigma\sigma' - \kappa\kappa').
\end{aligned}
\end{equation}
\item Suppose that $(\M,g_{ab},J^{a}{}_{b})$ is Hermitian. Let ${\rm T}_{s}$ be the spin-weight $s$ Teukolsky operator \eqref{TeukEq}. Then for any scalar field $\chi$ with weights $w=-s-1$, $p=2s$, we have
\begin{align}\label{RelationTeukBoxC}
 {\rm T}_{s}\chi = \left[ \Box^{\C} - 2(2s^2+1)\Psi_{2} \right]\chi.
\end{align}
\item Suppose that $(\M,g_{ab},J^{a}{}_{b})$ is conformally K\"ahler, where the K\"ahler metric is $\hat{g}_{ab}=\phi^2g_{ab}$. Then for any scalar field $\eta$ with weights $(w,0)$, we have
\begin{align}\label{identityCK}
 \Box^{\C}\eta= \phi^{-(w+1)}(\Box + 2\Psi_2 + \tfrac{R}{6})[\phi^{(w+1)}\eta].
\end{align}
\item If the geometry is conformally self-dual, then the operators \eqref{Dolbeault}, \eqref{BoxC} acting on any weighted scalar field satisfy
\begin{align}\label{CSD}
 \mathscr{C}^{A'}\tilde{\mathscr{C}}_{A'} = \tfrac{1}{2}\Box^{\mathscr{C}}, 
 \qquad 
 [\C_{A'},\tilde{\C}^{A'}]=0.
\end{align}
\end{enumerate}
\end{proposition}

\begin{proof}
For the first item, note that $\Box^{\C}=\C_{a}(g^{ab}\C_{b})$ and write the metric as $g^{ab}=2(\ell^a n^b - m^{a}\tilde{m}^{b}) - \i J^{ab}$, where $\ell^a,n^a,m^a,\tilde{m}^a$ is the null tetrad determined by the principal spinors $o^A,\iota^A$ of $J^{a}{}_{b}$ and an arbitrarily chosen primed spin frame (the operator \eqref{BoxC} is independent of this choice). From \eqref{Cdyad}, we always have $\C_{a}J^{ab}=0$. Expressions for $\C_{a}\ell^{a}$ and $\C_{a}m^{a}$ are easily obtained from \eqref{Cvector}, and for the contractions of $\C_{a}$ with the null tetrad one uses \eqref{Cscalar}.

The second item follows from \eqref{BoxCNP} by comparing this expression to \eqref{TeukEq} and noticing that on a Hermitian manifold, $\kappa=\sigma=0=\kappa'=\sigma'$ (cf. item \ref{item-integrableJ} in remark \ref{Rem:propertiesofC}).

For the third item, the proof is very similar to the proof of Lemma 3.2 in \cite{A18}, by using the fact that in a conformally K\"ahler manifold, the Lee form is $f_a=\partial_a\log\phi$.

Finally, the fourth item follows from equations (B.12), (B.13) and (B.26) in \cite{A21}, after noticing that a conformally self-dual geometry has $\Psi_2 = 0$.
\end{proof}

The above identities for the operator \eqref{BoxC} motivate our assertion in the introduction concerning the relation of eq. \eqref{BoxCIntro} to other equations in the literature. Since in this work we are interested in geometries with $\Psi_2=0$, we make the following definition:

\begin{definition}\label{Def:TeukEq}
Let $\chi$ be a weighted scalar field, with conformal weight $w$ and GHP weight $p$. For a conformally self-dual manifold $(\M,g_{ab})$, the second order equation 
\begin{align}
 \Box^{\C}\chi = 0 \label{TeukEqBoxC}
\end{align}
will be called the {\em Teukolsky equation for weights $(w,p)$}. 
\end{definition}

\subsection{The twistor distribution}
\label{Sec:TD}

For any Riemannian 4-manifold $(\mathcal{M},g_{ab})$, we define the correspondence space $\mathcal{F}$ to be the projective (unprimed) spin bundle, so that locally $\mathcal{F}$ is identified with $\M\times\CP^1$. A local section of $\mathcal{F}$ is a spinor field $\lambda^{A}$ up to scale. 
A choice of compatible almost-complex structure $J^{a}{}_{b}$ gives two preferred local trivialisations of $\mathcal{F}$, in the form of the spinor fields $o^{A},\iota^{A}$ subject to \eqref{GHP}. Any local section $\lambda^A$ can be written as 
\begin{align}\label{components}
\lambda^A = \lambda_{1}o^{A} - \lambda_{0}\iota^A,
\end{align}
where $\lambda_{0}=\lambda_{A}o^{A}$, $\lambda_{1}=\lambda_{A}\iota^{A}$. We can use $(\lambda_{0},\lambda_{1})$ as homogeneous coordinates in the fibres of $\mathcal{F}$. Note, however, that due to the freedom \eqref{GHP}-\eqref{CT}, the coordinates $(\lambda_{0},\lambda_{1})$ have weights: if $\lambda_{A}$ is conformally invariant, then 
\begin{align}\label{weightscoord}
 w(\lambda_{0})=-\frac{1}{2}, \quad p(\lambda_{0})=1,  
 \qquad 
 w(\lambda_{1})=-\frac{1}{2}, \quad p(\lambda_{1})=-1.
\end{align}

Any point $(x,\lambda_0,\lambda_1)$ in $\mathcal{F}$ defines a two-dimensional subspace $\{\lambda^{A}\partial_{AA'}\}\subset T_{x}\M$, called a twistor plane. Using the spin connection, these vectors can be lifted horizontally to $T_{(x,\lambda)}\mathcal{F}$; the resulting two-dimensional distribution in $\mathcal{F}$ is called the twistor distribution (cf. \cite[Section 13.3]{MW} and \cite[Section 10.1]{DunajskiBook}): it can be written abstractly as 
\begin{align}\label{TD}
 \lambda^{A}L_{AA'} = \lambda^{A}\partial_{AA'} - \Gamma_{AA'B}{}^{C}\lambda^{A}\lambda^{B}\tfrac{\partial}{\partial\lambda^{C}},
\end{align}
where $\Gamma_{AA'B}{}^{C}$ is the unprimed spin connection 1-form\footnote{Recall that given a spin frame $\varepsilon^{A}_{i}$, with dual frame $\theta^{i}_{A}$, the spin connection 1-form is the matrix-valued 1-form $\Gamma_{AA' i}{}^{j}$ defined by $\nabla_{AA'}\varepsilon^{C}_{i}=\Gamma_{AA' i}{}^{j}\varepsilon^{C}_{j}$. The object \eqref{spinconnection} is $\Gamma_{AA' B}{}^{C}:=\theta^{i}_{B}\nabla_{AA'}\varepsilon^{C}_{i}$ (sum over $i$).},
\begin{align}\label{spinconnection}
 \Gamma_{AA'B}{}^{C} = o_{B}\nabla_{AA'}\iota^{C} - \iota_{B}\nabla_{AA'}o^{C}.
\end{align}
The twistor distribution is integrable iff $(\M,g_{ab})$ is conformally self-dual (eq. \eqref{SDWeyl}, or in spinor terms $\Psi_{ABCD}\equiv 0$), see e.g. \cite[Proposition 13.4.1]{MW}. In this case, \eqref{TD} defines a foliation of the correspondence space $\mathcal{F}$ by integral surfaces. The space of leaves of the foliation is a three-dimensional complex manifold called twistor space: 
\begin{align}\label{Def:PT}
 \mathcal{PT} \equiv \mathcal{F} \diagup \{\lambda^{A}L_{AA'}\}.
\end{align}

The projection into $\M$ of each integral surface in $\mathcal{F}$ is a twistor surface, so $\mathcal{PT}$ can also be viewed as the moduli space of twistor surfaces in $\M$. Recalling the discussion below \eqref{SFR}, the existence of twistor surfaces defined by independent spinor fields implies that, locally, one can find integrable almost-complex structures on $\M$.

\section{Twistor Theory}
\label{Sec:TwistorTheory}

\subsection{The correspondence space}
\label{Sec:CorrespondenceSpace}

Let $(\M,g_{ab},J^{a}{}_{b})$ be almost-Hermitian, with the almost-complex structure given by \eqref{J}. Let $\mathcal{F}$ be the correspondence space as in section \ref{Sec:TD}. The local trivialisations defined by $o^A,\iota^A$ give a covering of $\mathcal{F}$ by two open sets, $U_{0}=\{\lambda_{0}\neq0\}$ and $U_{1}=\{\lambda_{1}\neq0\}$. The variables
\begin{align}\label{zeta}
\zeta=\lambda_{1}/\lambda_{0}, \qquad \tilde\zeta=\lambda_{0}/\lambda_{1},
\end{align} 
are inhomogeneous coordinates on the fibres over $U_0$ and over $U_1$ respectively. On the intersection $U_0\cap U_1$ we can use either of them; the relation is $\tilde\zeta=\zeta^{-1}$.

The tautological bundle $\mathcal{O}(-1)\to\mathcal{F}$ is the line bundle whose fibre over a point is the set of multiples of the corresponding projective spinor. It is determined by the transition function $\zeta$ from $U_1$ to $U_0$. Similarly, the bundles $\mathcal{O}(k)$, $k\in\mathbb{Z}$, are determined by the transition function $\zeta^{-k}$. A holomorphic section of $\mathcal{O}(k)$ is a function that is holomorphic and homogeneous of degree $k$ in $(\lambda_{0},\lambda_{1})$, that is $\Psi(x,c\lambda_{0},c\lambda_{1})=c^{k}\Psi(x,\lambda_{0},\lambda_{1})$ for any non-zero complex number $c$.
We are interested in holomorphic sections of $\mathcal{O}(k)$ over the intersection $U_0\cap U_1$; the set of such sections is denoted $\Gamma(U_0\cap U_1,\mathcal{O}(k))$.

In the trivialisation over $U_0$, we can write $\Psi(x,1,\zeta)=(\lambda_{0})^{-k}\Psi(x,\lambda_{0},\lambda_{1})$, which then gives $\Psi(x,\lambda_{0},\lambda_{1})=(\lambda_{0})^{k}f(x,\zeta)$, where $f(x,\zeta):=\Psi(x,1,\zeta)$. So $\Psi$ is represented by $f$ over $U_0$, and similarly it is represented by $\tilde{f}=\zeta^{-k}f$ over $U_1$. Noticing that $U_0\cap U_1$ is an annular region in $\CP^1$, and that the function $f$ is holomorphic there, it follows that $f$ admits a Laurent series expansion in $\zeta$:
\begin{align}\label{Laurent0}
 f(x,\zeta) = \sum_{j=-\infty}^{\infty}\psi_{j}(x)\zeta^{j}, 
 \qquad 
 \psi_{j}(x) = \frac{1}{2\pi\i}\oint_{\Gamma}\frac{f(x,\zeta)\d\zeta}{\zeta^{j+1}},
\end{align}
where the expression for $\psi_{j}$ is the standard formula for the coefficients of a Laurent series, see e.g. \cite{Sarason}, and $\Gamma$ is any contour in $\CP^1$ separating the points $\zeta=0$ and $\zeta=\infty$. (Since $f$ is holomorphic inside the annulus $U_0\cap U_1$, the integral is insensitive to continuous deformations of the contour.) So any $\Psi\in\Gamma(U_0\cap U_1,\mathcal{O}(k))$ can be expressed (in the trivialisation over $U_0$) as
\begin{align}\label{Laurent}
 \Psi(x,\lambda_{0},\lambda_{1}) 
 = (\lambda_0)^{k}\sum_{j=-\infty}^{\infty}\psi_{j}(x)\frac{(\lambda_{1})^{j}}{(\lambda_{0})^{j} }.
\end{align}

\begin{remark}
The coefficients $\psi_{j}$ are weighted scalar fields on $\M$: 
since $\Psi$ must have zero weights, using \eqref{weightscoord} and \eqref{Laurent} we deduce that the weights of $\psi_{j}$ are 
\begin{align}\label{weightscoefficients}
 w(\psi_{j}) = k/2, \qquad p(\psi_{j}) = 2j-k.
\end{align}
\end{remark}

The main result of this subsection is the following:

\begin{proposition}\label{prop:Hermitian}
Let $(\M,g_{ab},J^{a}{}_{b})$ be a Hermitian manifold, let $\lambda^{A}L_{AA'}$ be the twistor distribution \eqref{TD}, and let $\C_{AA'}$ be the connection defined in \eqref{DefC}. Then for any section $\Psi\in\Gamma(U_0\cap U_1,\mathcal{O}(k))$, with Laurent expansion \eqref{Laurent}, the following identity holds:
\begin{align}
 \lambda^{A}L_{AA'}\Psi = 
 (\lambda_0)^k \sum_{j=-\infty}^{\infty}\frac{(\lambda_{1})^{j}}{(\lambda_{0})^{j} }\lambda^{A}\C_{AA'}\psi_{j}.
 \label{TDH}
\end{align}
\end{proposition}

\begin{proof}
We have 
\begin{align*}
\lambda^{A}L_{AA'}\Psi = \sum_{j=-\infty}^{\infty}
 \left[ (\lambda_{0})^{k-j}(\lambda_{1})^{j}\lambda^A\partial_{AA'}\psi_{j}
  - \psi_{j}\Gamma_{AA'B}{}^{C}\lambda^A\lambda^B\tfrac{\partial}{\partial\lambda^C}
 \left((\lambda_{0})^{k-j}(\lambda_{1})^{j}\right)
 \right].
\end{align*}
From \eqref{components} we deduce $\frac{\partial\lambda_0}{\partial\lambda^C} = -o_C$,  $\frac{\partial\lambda_1}{\partial\lambda^C} = - \iota_C$, so we find that $\lambda^{A}L_{AA'}\Psi$ is equal to
\begin{align*}
\sum_{j=-\infty}^{\infty}(\lambda_{0})^{k-j}(\lambda_{1})^{j}
 \left[ \lambda^A\partial_{AA'}
 + \left( \frac{(k-j)}{\lambda_0}\Gamma_{AA'B}{}^{C}\lambda^A\lambda^Bo_C 
 + \frac{j}{\lambda_1}\Gamma_{AA'B}{}^{C}\lambda^A\lambda^B\iota_{C}
 \right) \right]\psi_{j}.
\end{align*}
To compute the connection terms, we express the spin connection in terms of the objects $f_{AA'},\omega_{AA'},\sigma^{0}_{A'},\sigma^{1}_{A'}$. Recalling \eqref{spinconnection}, a computation using the definitions \eqref{ConnectionForms}, \eqref{Cdyad} gives
\begin{subequations}\label{IdentitiesSpinConnection}
\begin{align}
 \Gamma_{AA'B}{}^{C}o_{C} &= o_{B}\omega_{AA'}+o_A\iota_Bo^Cf_{CA'} - \iota_A\iota_B\sigma^{0}_{A'}, \\
 \Gamma_{AA'B}{}^{C}\iota_{C} &= -\iota_{B}\omega_{AA'}-\iota_Ao_B\iota^Cf_{CA'} - o_Ao_B\sigma^{1}_{A'}.
\end{align}
\end{subequations}
Therefore: 
\begin{align*}
\Gamma_{AA'B}{}^{C}\lambda^A\lambda^Bo_C &= -(\lambda_1)^2\sigma^{0}_{A'} 
- \lambda_0\lambda_1o^A(\omega_{AA'}-f_{AA'}) + (\lambda_0)^2\iota^A\omega_{AA'}, \\
\Gamma_{AA'B}{}^{C}\lambda^A\lambda^B\iota_C &=  (\lambda_1)^2 o^A\omega_{AA'} 
- \lambda_0\lambda_1\iota^A(\omega_{AA'}+f_{AA'}) - (\lambda_0)^2\sigma^{1}_{A'}.
\end{align*}
So after a straightforward calculation we arrive at
\begin{equation*}
\begin{aligned}
\lambda^{A}L_{AA'}\Psi = \sum_{j=-\infty}^{\infty} & (\lambda_{0})^{k-j}(\lambda_{1})^{j}
\left\{  \lambda^A\partial_{AA'} + \lambda_1 o^A[(2j-k)\omega_{AA'}+(k-j) f_{AA'}]  \right. \\
& \left. -\lambda_0\iota^A[(2j-k)\omega_{AA'}+jf_{AA'}] 
+ \tfrac{(j-k)(\lambda_1)^2}{\lambda_0}\sigma^0_{A'} - \tfrac{j(\lambda_0)^2}{\lambda_1}\sigma^1_{A'} \right\}\psi_{j}.
\end{aligned}
\end{equation*}
For an arbitrary scalar field $\psi$ with weights $(w,p)$, using \eqref{components}, \eqref{Cscalar} and \eqref{ConnectionForms} we get
\begin{align*}
\lambda^{A}\C_{AA'}\psi = \{\lambda^{A}\partial_{AA'} 
+ \lambda_{1}o^{A}[p\omega_{AA'}+(w-\tfrac{p}{2})f_{AA'}]
- \lambda_{0}\iota^{A}[p\omega_{AA'}+(w+\tfrac{p}{2})f_{AA'}]\}\psi.
\end{align*}
From remark \ref{Rem:propertiesofC}, on a Hermitian manifold we have $\sigma^0_{A'}=\sigma^1_{A'} = 0$. 
Recalling now that $\psi_{j}$ has weights as in \eqref{weightscoefficients}, the result follows.
\end{proof}

\subsection{Twistor functions and recursion relations}
\label{Sec:RecursionRelations}

We now assume that $(\M,g_{ab})$ is conformally self-dual, eq. \eqref{SDWeyl}. From section \ref{Sec:TD}, the twistor distribution \eqref{TD} is integrable, and there is a twistor space $\mathcal{PT}$ \eqref{Def:PT}. Similarly to $\mathcal{F}$, this can be covered by two open sets $V_0, V_1$, which we may take to be the images of $U_0, U_1$ under the projection $\mathcal{F}\to\mathcal{PT}$. The analogue of the bundles $\mathcal{O}(k)$ over $\mathcal{F}$ but now over $\mathcal{PT}$ will also be denoted by $\mathcal{O}(k)$ (those over $\mathcal{F}$ are the pull-back of those over $\mathcal{PT}$). A local section $\Psi\in\Gamma(V_0\cap V_1,\mathcal{O}(k))$ can be thought of as a local section\footnote{We slightly abuse notation and use the same letter $\Psi$ for both spaces.} $\Psi\in\Gamma(U_0\cap U_1,\mathcal{O}(k))$ which is constant along the twistor distribution:
\begin{align}\label{TF}
 \lambda^{A}L_{AA'}\Psi=0.
\end{align}
We call such functions {\em twistor functions}.

\begin{proposition}\label{prop:TF1}
Let $\Psi\in\Gamma(V_0\cap V_1,\mathcal{O}(k))$ be a twistor function. Then the coefficients $\psi_{j}$ in the associated Laurent series \eqref{Laurent} solve the Teukolsky equation for weights $(\frac{k}{2},2j-k)$, eq. \eqref{TeukEqBoxC}. The coefficients can be recovered from the twistor function $\Psi$ by a contour integral:
\begin{align}\label{ContourIntegral}
 \psi_{j}(x) = \frac{1}{2\pi\i}\oint_{\Gamma}\frac{\Psi(x,\lambda)\lambda_{A}\d\lambda^A}{(o_B\lambda^B)^{k-j+1}(\iota_B\lambda^B)^{j+1}}.
\end{align}
\end{proposition}

\begin{proof}
From proposition \ref{prop:Hermitian}, we see that if $\Psi$ satisfies \eqref{TF} then the right hand side of \eqref{TDH} must vanish. Using \eqref{components}, it follows that $o^{A}\C_{AA'}\psi_{j} - \iota^{A}\C_{AA'}\psi_{j+1}=0$ for all $j$, or equivalently using the operators \eqref{Dolbeault}:
\begin{align}\label{RecursionRelations}
 \tilde{\C}_{A'}\psi_{j} = \C_{A'}\psi_{j+1} \qquad \forall \, j.
\end{align}
Applying $\C^{A'}$ to this equation and using \eqref{DolbeaultSquared}, we see that $\C^{A'}\tilde{\C}_{A'}\psi_{j}=0$. Recalling now \eqref{CSD}, we get equation \eqref{TeukEqBoxC}. The  integral \eqref{ContourIntegral} follows from the expression for the coefficients in \eqref{Laurent0}, noticing that $f(x,\zeta)=(\lambda_{0})^{-k}\Psi(x,\lambda_0,\lambda_1)$ and $\lambda_A\d\lambda^A=(\lambda_0)^2\d(\frac{\lambda_1}{\lambda_0})$.
\end{proof}

Following the terminology of the twistor literature \cite{MW, DM}, we call \eqref{RecursionRelations} the {\em recursion relations}. They are connected to a symmetry operator for the Teukolsky equations, which we call the {\em recursion operator}:
\begin{proposition}\label{prop:RO}
Let $\mathcal{T}_{(w,p)}$ be the space of solutions to the Teukolsky equation for weights $(w,p)$. Then there is an operator $\mathcal{R}:\mathcal{T}_{(w,p)}\to \mathcal{T}_{(w,p+2)}$ defined, for any $\chi\in\mathcal{T}_{(w,p)}$, by 
\begin{align}
 \tilde{\C}_{A'}\chi = \C_{A'}\mathcal{R}\chi. \label{RecursionOp}
\end{align}
\end{proposition}

\begin{proof}
Let $\chi$ be a solution to the Teukolsky equation $\Box^{\C}\chi=0$ for weights $(w,p)$, so $\chi\in\mathcal{T}_{(w,p)}$. Using the first identity in \eqref{CSD}, this can be written as $\C^{A'}\tilde{\C}_{A'}\chi=0$. From item \ref{item:potentials} in remark \ref{Rem:propertiesofC}, it follows that there exists a scalar field $\tilde\chi$ such that $\tilde{\C}_{A'}\chi=\C_{A'}\tilde\chi$. The weights of $\tilde\chi$ are $(w,p+2)$. Applying $\tilde{\C}^{A'}$ and using \eqref{DolbeaultSquared}, we get $\tilde{\C}^{A'}\C_{A'}\tilde\chi=0$. Using the second and first identities in \eqref{CSD}, we have $\Box^{\C}\tilde\chi=0$, so $\tilde\chi\in\mathcal{T}_{(w,p+2)}$. We then set $\tilde\chi\equiv\mathcal{R}\chi$.
\end{proof}

With the help of the recursion operator $\mathcal{R}$, we have a converse to proposition \ref{prop:TF1}:
\begin{proposition}\label{prop:TF2}
Let $\chi$ be a solution to the Teukolsky equation for weights $(w,p)$, eq. \eqref{TeukEqBoxC}. Then $\chi$ generates a twistor function $\Psi\in\Gamma(V_0\cap V_1,\mathcal{O}(2w))$, and we have
\begin{align}\label{ContourIntegral2}
 \chi(x) = \frac{1}{2\pi\i}\oint_{\Gamma}\frac{\Psi(x,\lambda)\lambda_{A}\d\lambda^A}{(o_B\lambda^B)^{w+1-\frac{p}{2}}(\iota_B\lambda^B)^{w+1+\frac{p}{2}}}.
\end{align}
\end{proposition}

\begin{proof}
Given $\chi$, using the recursion operator \eqref{RecursionOp} we get a solution $\chi^{(1)}=\mathcal{R}\chi$ to the Teukolsky equation for weights $(w,p+2)$. Since $\tilde{\C}^{A'}\C_{A'}\chi^{(1)}=0$ and the operators commute, we can write $\C^{A'}\tilde{\C}_{A'}\chi^{(1)}=0$. Applying \eqref{potentials} again, there exists $\chi^{(2)}$ such that $\tilde{\C}_{A'}\chi^{(1)}=\C_{A'}\chi^{(2)}$, and we have $\chi^{(2)}=\mathcal{R}(\mathcal{R}\chi)\equiv\mathcal{R}^2\chi$. The weights of $\chi^{(2)}$ are $(w,p+4)$. Iterating this process we get a set $\{\mathcal{R}^{m}\chi\}_{m=0}^{\infty}$, where the weights of $\mathcal{R}^{m}\chi$ are $(w,p+2m)$. To get negative $m$'s, we start again from the Teukolsky equation for $\chi$ but now we write it as $\tilde{\C}^{A'}\C_{A'}\chi=0$, so there exists $\chi^{(-1)}$ such that $\C_{A'}\chi=\tilde{\C}_{A'}\chi^{(-1)}$. Comparing to \eqref{RecursionOp}, we can write $\chi=\mathcal{R}\chi^{(-1)}$, so formally $\chi^{(-1)}=\mathcal{R}^{-1}\chi$. Repeating the process, we end up with an infinite family, $\{\mathcal{R}^{m}\chi\}_{m=-\infty}^{\infty}$, of solutions to the Teukolsky equations, connected by the recursion operator \eqref{RecursionOp} or, equivalently, by the recursion relations $\tilde{\C}_{A'}\mathcal{R}^{m}\chi=\C_{A'}\mathcal{R}^{m+1}\chi$. We can then define the power series
\begin{align}\label{GeneratedTF}
\Psi(x,\lambda_{0},\lambda_{1}) := \sum_{m=-\infty}^{\infty}(\mathcal{R}^{m}\chi)(x)(\lambda_{0})^{w-\frac{p}{2}-m}(\lambda_{1})^{w+\frac{p}{2}+m}.
\end{align}
The function $\Psi$ has weights $(0,0)$, it is homogeneous of degree $2w$ in $(\lambda_0,\lambda_1)$, and, since the coefficients $\mathcal{R}^{m}\chi$ satisfy the recursion relations, we have $\lambda^{A}L_{AA'}\Psi=0$. Thus $\Psi$ is a twistor function in $\Gamma(V_0\cap V_1,\mathcal{O}(2w))$. The integral \eqref{ContourIntegral2} is simply \eqref{ContourIntegral} for $k=2w$, $j=w+\frac{p}{2}$.
\end{proof}

Proposition \ref{prop:TF2} tells us that any solution to the Teukolsky equation generates a twistor function, and proposition \ref{prop:RO} tells us that, given a solution $\chi$ to the Teukolsky equation for weights $(w,p)$, the recursion operator produces a solution $\tilde\chi=\mathcal{R}\chi$ to the equation for weights $(w,p+2)$. It is then natural to ask about the relation between the twistor functions generated by $\chi$ and by $\mathcal{R}\chi$. We have:

\begin{proposition}\label{prop:ROTF}
$\chi$ and $\mathcal{R}\chi$ generate the same twistor function.
\end{proposition}

\begin{proof}
Let $\chi$ be a solution to the Teukolsky equation for weights $(w,p)$. The generated twistor function, say $\Psi$, is \eqref{GeneratedTF}. Likewise, the twistor function corresponding to $\tilde\chi=\mathcal{R}\chi$, say $\tilde\Psi$, is given by the same expression \eqref{GeneratedTF} but replacing $\chi$ by $\tilde\chi$ and $p$ by $p+2$, since the weights of $\tilde\chi$ are $(w,p+2)$. Thus:
\begin{align*}
\tilde\Psi(x,\lambda) &= \sum_{\ell=-\infty}^{\infty}(\mathcal{R}^{\ell}\tilde\chi)(x)(\lambda_{0})^{w-\frac{(p+2)}{2}-\ell}(\lambda_{1})^{w+\frac{(p+2)}{2}+\ell} \\
& = \sum_{\ell=-\infty}^{\infty}(\mathcal{R}^{\ell+1}\chi)(x)(\lambda_{0})^{w-\frac{p}{2}-\ell-1}(\lambda_{1})^{w+\frac{p}{2}+\ell+1}.
\end{align*}
The second line is simply \eqref{GeneratedTF} with $m=\ell+1$, so $\tilde\Psi=\Psi$.
\end{proof}

\subsection{Twistor cohomology}
\label{Sec:Cohomology}

From prop. \ref{prop:TF2}, any solution $\chi$ to the Teukolsky equation for weights $(w,p)$ can be expressed as the contour integral of a twistor function $\Psi$ as in \eqref{ContourIntegral2}. Suppose we make a transformation
\begin{align}\label{freedom}
 \Psi \to \Psi' = \Psi + h_0 + h_1,
\end{align}
where we assume that $h_i$ is a function such that its quotient by the denominator in \eqref{ContourIntegral2} is holomorphic in $V_i$, $i=0,1$. Then the result of the integral is unchanged. So, from the perspective of the contour integral, $\chi$ would generate not just a single twistor function but a class of them related by \eqref{freedom}, as long as we can find the required functions $h_0,h_1$.

To investigate whether such functions can be found, we will use sheaf cohomology, as is usual in twistor constructions. We will follow the \v{C}ech approach, see e.g. \cite{HuggettTod, MW, PR2, WW}. The first \v{C}ech cohomology group of $\PT$ with coefficients in $\mathcal{O}(k)$, denoted $\check{H}^{1}(\PT,\mathcal{O}(k))$, is the set of 1-cocycles modulo 1-coboundaries. In our case we only have two sets in the cover, so the cocycle condition is trivial, i.e. every cochain is a cocycle. Following \cite{MW} we can then write\footnote{Recall that $\Gamma(\mathcal{U},E)$ is the set of holomorphic sections of a bundle $E$ over an open set $\mathcal{U}$.} 
\begin{align}\label{H1}
\check{H}^{1}(\PT,\mathcal{O}(k)) = \frac{\Gamma(V_{0}\cap V_{1},\mathcal{O}(k))}{\Gamma(V_{0},\mathcal{O}(k))+\Gamma(V_{1},\mathcal{O}(k))}.
\end{align}
A cohomology class in $\check{H}^{1}(\PT,\mathcal{O}(k))$ is thus represented by a cochain, which we recall is a local holomorphic section of $\mathcal{O}(k)$ over $V_0\cap V_1$. So we can say that any such section, say $F$, is a \v{C}ech representative in $\check{H}^{1}(\PT,\mathcal{O}(k))$. The quotient in \eqref{H1} means that $F$ and $F'\equiv F + h_{0} + h_{1}$, where $h_{0}$ is holomorphic in $V_0$ and $h_{1}$ is holomorphic in $V_1$, represent the same cohomology class. Elements of the same class are said to be cohomologous. 

With the above terminology, we can replace the term `twistor function' in propositions \ref{prop:TF1} and \ref{prop:TF2} by `\v{C}ech representative'. We can also investigate the conditions under which $\chi$ generates not just a \v{C}ech representative but a cohomology class: that is, given $\chi$, we say that it ``generates a cohomology class'' if it is invariant under \eqref{freedom}.

\begin{proposition}\label{prop:cohomology}
Let $\chi$ be a solution to the Teukolsky equation for weights $(w,p)$. Then $\chi$ generates a cohomology class in $\check{H}^{1}(\PT,\mathcal{O}(2w))$ if and only if $w$ and $p$ satisfy
\begin{align}\label{Restrictionspw}
 w \pm \frac{p}{2} \in - \mathbb{N}.
\end{align}
\end{proposition}

\begin{proof}
Let $h_i\in\Gamma(V_i,\mathcal{O}(2w))$, with $i=0,1$. Put $\Psi'=\Psi+h_0+h_1$, then $\Psi'$ and $\Psi$ are cohomologous. Letting $\chi'$ be the space-time field generated by inserting $\Psi'$ in the integrand in the right hand side of \eqref{ContourIntegral2}, we wish to know whether $\chi'$ and $\chi$ are equal for any choice of $h_{0},h_{1}$. Write $h_0,h_1$ in terms of their Taylor series:
\begin{align}\label{Taylor}
 h_{0}(x,\lambda) = (\lambda_0)^{2w}\sum_{j=0}^{\infty}a_{j}(x)\frac{(\lambda_{1})^{j}}{(\lambda_{0})^{j}}, 
 \qquad 
 h_{1}(x,\lambda) = (\lambda_1)^{2w}\sum_{j=0}^{\infty}b_{j}(x)\frac{(\lambda_{0})^{j}}{(\lambda_{1})^{j}}.
\end{align}
Replacing in the contour integral \eqref{ContourIntegral2} for $\Psi'$, we find
\begin{align}
\nonumber \chi' &= \chi + \frac{1}{2\pi\i}\oint_{\Gamma}\frac{h_{0}(x,\lambda)\lambda_A\d\lambda^A}{(\lambda_0)^{w+1-\frac{p}{2}} (\lambda_1)^{w+1+\frac{p}{2}}} 
+ \frac{1}{2\pi\i}\oint_{\Gamma}\frac{h_{1}(x,\lambda)\lambda_A\d\lambda^A}{(\lambda_0)^{w+1-\frac{p}{2}} (\lambda_1)^{w+1+\frac{p}{2}}} \\
&= \chi + \sum_{j=0}^{\infty}\frac{a_j}{2\pi\i}\oint_{\Gamma}\zeta^{j-w-1-\frac{p}{2}}\d\zeta 
 - \sum_{j=0}^{\infty}\frac{b_j}{2\pi\i}\oint_{\Gamma}\tilde\zeta^{j-w-1+\frac{p}{2}}\d\tilde\zeta 
 \label{chi'}
\end{align}
where in the second line we replaced \eqref{Taylor} together with $\lambda_A\d\lambda^A=(\lambda_0)^2\d\zeta=-(\lambda_1)^2\d\tilde\zeta$ (recall \eqref{zeta}). The functions $h_0,h_1$ are arbitrary, so we see that $\chi'$ and $\chi$ will be equal iff the terms with $a_j$ and $b_j$ in \eqref{chi'} vanish independently. The terms with $a_j$ vanish iff $j-w-1-\frac{p}{2}$ is a non-negative integer for all $j$, and the terms with $b_{j}$ vanish iff $j-w-1+\frac{p}{2}$ is a non-negative integer for all $j$. Since from \eqref{Taylor} we have $j\geq0$, this is equivalent to 
\begin{align*}
 w+1\pm\frac{p}{2} \in -\mathbb{N}\cup0,
\end{align*}
which can be rewritten as \eqref{Restrictionspw}.
\end{proof}

We now wish to give an interpretation of the restrictions \eqref{Restrictionspw}. Notice first that if $w\geq-\frac{1}{2}$, then there is no $p$ such that \eqref{Restrictionspw} can be satisfied. So we must have $w\leq-1$. In this case, suppose that $\chi_{(w,p)}$ is a solution to the Teukolsky equation for weights $(w,p)$, where $w\pm\frac{p}{2}\in-\mathbb{N}$. Applying the recursion operator $\mathcal{R}$ and its inverse, we generate other solutions $\chi_{(w,p_{j})}=\mathcal{R}^{j}\chi_{(w,p)}$, with $p_{j}=p+2j$, $j\in\mathbb{Z}$. A certain number of the new solutions will still satisfy the constraint $w\pm\frac{p_{j}}{2}\in-\mathbb{N}$. To find this number, notice that the conditions $w\pm\frac{p_{j}}{2}\in-\mathbb{N}$ are equivalent to $p_{j}=-2w-2n$ and $p_{j}=2w+2\tilde{n}$ for $n,\tilde{n}\in\mathbb{N}$, so the possible values of $p_{j}$ are $2w+2, 2w+4, 2w+6$, $...$, $-2w-6, -2w-4, -2w-2$. The number of different $p$'s is then $-2w-1$. So we have $-2w-1$ fields $\{\chi_{(w,p_j)}\}_{j=0}^{-2w-2}$, where each $\chi_{(w,p_j)}$ solves the Teukolsky equation for weights $(w,p_{j})$, and since $(w,p_j)$ satisfy \eqref{Restrictionspw}, each $\chi_{(w,p_j)}$ generates a cohomology class in $\check{H}^{1}(\PT,\mathcal{O}(2w))$. But since these fields are related by the recursion operator, proposition \ref{prop:ROTF} implies that they generate the {\em same} twistor function $\Psi$, and so the {\em same} cohomology class $[\Psi]\in\check{H}^{1}(\PT,\mathcal{O}(2w))$. In summary, the restrictions \eqref{Restrictionspw} mean that a single class in $\check{H}^{1}(\PT,\mathcal{O}(2w))$ is associated to $N$ solutions to the Teukolsky equations, where $N=0$ if $w>-1$ and $N=-2w-1$ if $w\leq-1$. But now we see that this is simply the known structure of the cohomology group $\check{H}^{1}(\PT,\mathcal{O}(2w))$: it is well-known (see e.g. \cite{HuggettTod, MW}) that, on restriction to a twistor line $\CP^1$, one has
\begin{align}\label{H1CP1}
 \check{H}^{1}(\CP^1,\mathcal{O}(2w)) = 
 \begin{cases} \mathbb{C}^{-2w-1} \quad &\text{if } \quad w\leq-1 \\
 0 \quad &\text{if } \quad w>-1 \end{cases}.
\end{align}

\subsection{The conformally K\"ahler case}
\label{Sec:CK}

Suppose that the geometry satisfies \eqref{SDWeyl} and is also conformal to a K\"ahler metric (which is then a scalar-flat K\"ahler surface). In our framework this means that the Lee form is locally exact, $f_a=\partial_a\log\phi$ for some $\phi$. The interesting point now is that, by using the recursion operator $\mathcal{R}$, we can map the Teukolsky equation with any weights $(w,p)$ to the conformal Laplacian, as long as $p$ is even.

\begin{proposition}\label{prop:CKcase}
Let $(\M,g_{ab})$ be conformal to a scalar-flat K\"ahler surface. Let $\chi_{(w,p)}$ be an arbitrary weighted scalar field with weights $(w,p)$, where $p$ is even. Then we have the identity
\begin{align}\label{identityCK3}
\Box^{\C}\mathcal{R}^{-\frac{p}{2}}\chi_{(w,p)} = \phi^{-(w+1)}\left(\Box + \tfrac{R}{6} \right)\left[\phi^{w+1}\mathcal{R}^{-\frac{p}{2}}\chi_{(w,p)} \right].
\end{align}
\end{proposition}

\begin{proof}
Recall that if $\chi_{(w,p)}$ has weights $(w,p)$, then $\mathcal{R}^{m}\chi_{(w,p)}$ has weights $(w,p+2m)$, with $m\in\mathbb{Z}$. If $p$ is even, it follows that $\mathcal{R}^{-\frac{p}{2}}\chi_{(w,p)}$ has weights $(w,0)$. Then the identity \eqref{identityCK} for the Teukolsky operator $\Box^{\C}$ acting on scalar fields with weights $(w,0)$ applies to $\mathcal{R}^{-\frac{p}{2}}\chi_{(w,p)}$. But since $\Psi_2\equiv0$, we see that the operator simply reduces to the conformal Laplacian composed with conjugation by $\phi^{w+1}$.
\end{proof}

If $\chi_{(w,p)}$ is now a solution to the Teukolsky equation for weights $(w,p)$, then $\mathcal{R}^{-\frac{p}{2}}\chi_{(w,p)}$ is a solution to the corresponding equation for weights $(w,0)$, since $\mathcal{R}$ maps solutions to solutions. Then the scalar field $\phi^{w+1}\mathcal{R}^{-\frac{p}{2}}\chi_{(w,p)}$ has weights $(-1,0)$ and, in view of \eqref{identityCK3}, is a solution to the conformal wave equation. Conversely, suppose that $\Phi$ has weights $(-1,0)$ and solves the conformal wave equation $(\Box+\frac{R}{6})\Phi=0$. For $p$ even, the field $\chi_{(w,p)}:=\phi^{-(w+1)}\mathcal{R}^{\frac{p}{2}}\Phi$ has weights $(w,p)$ and, again as a consequence of \eqref{identityCK3}, solves $\Box^{\C}\chi_{(w,p)}=0$, which is the Teukolsky equation for weights $(w,p)$.

\section{Quaternionic-K\"ahler Metrics}
\label{Sec:QK}

Here we specialise to the case in which the geometry $(\M,g_{ab})$ satisfies not only \eqref{SDWeyl} but is also Einstein, with non-zero cosmological constant $\lambda$. In four dimensions, these manifolds are called quaternionic-K\"ahler (or self-dual Einstein)\footnote{Despite the name, these manifolds are generically not K\"ahler.}. The curvature satisfies
\begin{align}\label{qK}
 \Psi_{ABCD}=0, \qquad R_{ab}=\lambda g_{ab}, \qquad \lambda\neq0.
\end{align}

\subsection{Gravitational perturbations}
\label{Sec:GPqK}

We will focus on metric deformations, that is, we perturb the metric $g_{ab} \to g_{ab}+\varepsilon h_{ab}$, and we wish to study the linearised Einstein equations $\mathcal{E}_{ab}[h]=0$, where
\begin{align}\label{LEE}
\mathcal{E}_{ab}[h]:=
-\tfrac{1}{2}\Box h_{ab} - \tfrac{1}{2}\nabla_a\nabla_b{\rm tr}(h) + \nabla^{c}\nabla_{(a}h_{b)c} - \tfrac{1}{2}g_{ab}[\nabla^c\nabla^d h_{cd} - \Box{\rm tr}(h)] + \lambda h_{ab} = 0
\end{align}
and $\Box=g^{ab}\nabla_a\nabla_b$, ${\rm tr}(h)=g^{ab}h_{ab}$.

\begin{proposition}\label{prop:SEOT}
Let $(\M,g_{ab})$ be quaternionic-K\"ahler \eqref{qK}. Fix an integrable complex structure, with principal spinors $o_A,\iota_A$, Lee form $f_{a}$, and connection $\C_a$ as in section \ref{Sec:preliminaries}. Then for an arbitrary symmetric tensor field $h_{ab}$, the following identities hold:
\begin{align}\label{SEOT}
 - \mathring\Omega^{-1} W^{abcd}_{\bf i}(\nabla_a-4f_a)\nabla_d\mathcal{E}_{bc}[h] 
  = \tfrac{1}{4}\Box^{\C}\left(\mathring\Omega^{-1}W^{abcd}_{\bf i}\dot{C}_{abcd}[h]\right), 
  \qquad {\bf i}=0,...,4
\end{align}
where $\mathring{\Omega}$ is an auxiliary constant conformal factor, $\mathcal{E}_{ab}[h]$ is the linearised Einstein operator \eqref{LEE}, $\dot{C}_{abcd}[h]$ is the linearised Weyl tensor thought of as a linear functional of $h_{ab}$, and 
\begin{align}\label{Wtensors}
W_{\bf i}^{abcd}:=S_{\bf i}^{ABCD}\epsilon^{A'B'}\epsilon^{C'D'}, 
\end{align}
with $S_{0}^{ABCD}=o^{(A}o^Bo^Co^{D)}$, $S_{1}^{ABCD}=o^{(A}o^Bo^C\iota^{D)}$, ..., $S_{4}^{ABCD}=\iota^{(A}\iota^B\iota^C\iota^{D)}$.
\end{proposition}

\begin{proof}
Suppose first that $(\M,g_{ab},J^{a}{}_{b})$ is an arbitrary almost-Hermitian manifold, with ASD Weyl spinor $\Psi_{ABCD}$. Recalling that $\Psi_{ABCD}$ has conformal weight zero, let $\mathring{\Omega}$ be an (auxiliary) constant conformal factor, $\nabla_{a}\mathring{\Omega}=0$, and define 
\begin{align}\label{spin2field}
 \varphi_{ABCD}:=\mathring{\Omega}^{-1}\Psi_{ABCD},
\end{align}
which has conformal weight $w=-1$ (and $p=0$). Using \eqref{Cspinor}, we have: 
\begin{align*}
 S_{\bf i}^{ABCD}\C_{AA'}\C^{A'E}\varphi_{BCDE} 
 &= \mathring{\Omega}^{-1}S_{\bf i}^{ABCD}(\nabla_{AA'}-4f_{AA'})\nabla^{A'E}\Psi_{BCDE}\\
 &= \mathring{\Omega}^{-1}S_{\bf i}^{ABCD} (\nabla_{AA'}-4f_{AA'}) \nabla_{B}^{B'} 
 \Phi_{CDB'}{}^{A'} \\
 &=\mathring{\Omega}^{-1}S_{\bf i}^{ACBD}\epsilon^{A'C'}\epsilon^{B'D'}(\nabla_{AA'}-4f_{AA'})\nabla_{DD'}\Phi_{BCB'C'} \\
 &=  - \tfrac{1}{2}\mathring\Omega^{-1} W^{abcd}_{\bf i}(\nabla_a-4f_a)\nabla_{d}(G_{bc}+\tfrac{R}{4}g_{bc}),
\end{align*}
where in the second line we used Bianchi identities \cite[Eq. (4.10.7)]{PR1}, and in the last line we used the definition \eqref{Wtensors} and the relation between the Ricci spinor $\Phi_{ABA'B'}$ and the Einstein tensor $G_{ab}$ \cite[Eq. (4.6.25)]{PR1}. We now linearise the above equation around a space that satisfies \eqref{qK}. The ASD linearised Weyl tensor has a spinor decomposition $\dot{C}^{-}_{abcd}\equiv\dot\Psi_{ABCD}\epsilon_{A'B'}\epsilon_{C'D'}$, and correspondingly we have a field $\dot{\varphi}_{ABCD}$ as in \eqref{spin2field}. The fields $S_{\bf i}^{ABCD}$ and the operator $\C_{AA'}$ can now be taken to be the ones of the background, so we have
\begin{align*}
 S_{\bf i}^{ABCD}\C_{AA'}\C^{A'E}\dot{\varphi}_{BCDE} 
 = \tfrac{1}{2}\Box^{\C}(S_{\bf i}^{ABCD}\dot{\varphi}_{ABCD}),
\end{align*}
thus \eqref{SEOT} follows.
\end{proof}

We see that the role of the auxiliary constant conformal factor $\mathring{\Omega}$ is simply to allow us to apply the operator $\C_{AA'}$ to fields with the appropriate conformal weights; once the final formulae have been obtained, $\mathring{\Omega}$ can simply be set to one.

The point of the identity \eqref{SEOT} is the following. Suppose first that $h_{ab}$ is a solution to the linearised Einstein equations, $\mathcal{E}_{ab}[h]=0$. Define the weighted scalar fields
\begin{align}\label{TeukFields}
 \chi_{\bf i}:=\tfrac{1}{4}\mathring\Omega^{-1}W^{abcd}_{\bf i}\dot{C}_{abcd}[h],
\end{align}
whose weights are $w=-3$, $p=4-2{\bf i}$. Then \eqref{SEOT} tells us that each $\chi_{\bf i}$, ${\bf i}=0,...,4$, is a solution to the Teukolsky equation for weights $(-3,4-2{\bf i})$. From proposition \ref{prop:cohomology}, since the weights satisfy \eqref{Restrictionspw}, each $\chi_{\bf i}$ then generates a cohomology class in $\check{H}^{1}(\PT,\mathcal{O}(-6))$. In addition, the linearised Bianchi identity $\nabla^{AA'}\dot\Psi_{ABCD}=0$ implies $\C^{AA'}\dot{\varphi}_{ABCD}=0$, which gives $\tilde\C_{A'}(\iota^A\dot{\varphi}_{ABCD})=\C_{A'}(o^A\dot{\varphi}_{ABCD})$. Contracting with $o^A$'s and $\iota^A$'s, we see that the $\chi_{\bf i}$'s satisfy the recursion relations, therefore, from prop. \ref{prop:ROTF}, they generate the same class. In summary, a solution $h_{ab}$ to the linearised Einstein equations generates a cohomology class in $\check{H}^{1}(\PT,\mathcal{O}(-6))$ via the perturbed curvature.

In the opposite direction, from prop. \ref{prop:TF1} we can use a representative in $\check{H}^{1}(\PT,\mathcal{O}(-6))$ to generate solutions to the Teukolsky equations; however, how do we reconstruct a perturbed metric from them? The identity \eqref{SEOT} gives a way to do this, but it also suggests that it is actually the group $\check{H}^{1}(\PT,\mathcal{O}(+2))$ which is naturally involved here:
\begin{proposition}\label{prop:metricpert}
Let $\Phi^{\bf i}$ be a solution to the Teukolsky equation with weights $w=1$, $p=2{\bf i}-4$, for any ${\bf i}$ in $\{0,...,4\}$. Then the tensor field
\begin{align}\label{metricpert}
 h_{ab}^{\bf i} := -\nabla^{d}[(\nabla^{c}+4f^{c})W^{\bf i}_{c(ab)d}\Phi^{\bf i}]
\end{align}
(no sum over $\bf i$) is a solution to the linearised Einstein equations, $\mathcal{E}_{ab}[h^{\bf i}]=0$. 
\end{proposition}

\begin{proof}
We use Wald's adjoint operators method \cite{Wald2} and write \eqref{SEOT} as the operator identity $\mathcal{S}_{\bf i}\mathcal{E}[h_{ab}]=\mathcal{O}\mathcal{T}_{\bf i}[h_{ab}]$ valid for any symmetric $h_{ab}$, where: $\mathcal{E}$ is defined in \eqref{LEE}, $\mathcal{O}(\chi):=\Box^{\C}\chi$, $\mathcal{T}_{\bf i}[h_{ab}]:=\chi_{\bf i}$ (with $\chi_{\bf i}$ as in \eqref{TeukFields}), and $\mathcal{S}_{\bf i}$ is the operator that is acting on $\mathcal{E}_{bc}[h]$ in the left hand side of \eqref{SEOT}. The adjoint identity is $\mathcal{E}\mathcal{S}^{\dagger}_{\bf i}[\Phi]=\mathcal{T}^{\dagger}_{\bf i}\mathcal{O}^{\dagger}[\Phi]$, valid for any scalar field $\Phi$ with the appropriate weights, where we used that $\mathcal{E}$ is self-adjoint \cite{Wald2}. A straightforward calculation shows that $\mathcal{S}^{\dagger}_{{\bf i} \, ab}\equiv h^{\bf i}_{ab}$, with $h^{\bf i}_{ab}$ given by \eqref{metricpert}. It follows that if $\mathcal{O}^{\dagger}[\Phi^{\bf i}]=0$, then \eqref{metricpert} is a solution to the linearised Einstein equations. So it remains to understand the operator $\mathcal{O}^{\dagger}$. To this end, we recall that the adjoint of $\mathcal{O}$ is the unique operator $\mathcal{O}^{\dagger}$ defined by 
\begin{align*}
 \Phi^{\bf i}\mathcal{O}[\chi_{\bf i}] = \chi_{\bf i}\mathcal{O}^{\dagger}[\Phi^{\bf i}] + \nabla_{a}v^{a}_{\bf i}
 \qquad \text{(no sum over $\bf i$)}
\end{align*}
for some vector field $v^a_{\bf i}$, cf. \cite[Eq. (9)]{Wald2}. Since $\chi_{\bf i}$ has weights $(-3,4-2{\bf i})$, we have that $\mathcal{O}[\chi_{\bf i}]=g^{ab}\C_a\C_b\chi_{\bf i}$ has weights $(-5,4-2{\bf i})$. Multiplying by $\Phi^{\bf i}$, we get $\Phi^{\bf i}\mathcal{O}[\chi_{\bf i}]=\C_av^a_{\bf i}+(\Box^{\C}\Phi^{\bf i})\chi_{\bf i}$, where $v^a_{\bf i}=\Phi^{\bf i}\C^a\chi_{\bf i} - \chi_{\bf i}\C^a\Phi^{\bf i}$. Using \eqref{Cvector}, it is straightforward to show that in order to have $\C_av^a_{\bf i}=\nabla_av^a_{\bf i}$, the weights of $\Phi^{\bf i}$ must be $w=+1$, $p=2{\bf i}-4$. So we see that for $\mathcal{O}=\Box^{\C}$ as defined by \eqref{SEOT} (that is, acting on scalars with weights $(-3,4-2{\bf i})$), the adjoint is $\mathcal{O}^{\dagger}=\Box^{\C}$ but acting on scalars with weights $(1,2{\bf i}-4)$.
\end{proof}

Since in the above proposition $\Phi^{\bf i}$ is a solution to the Teukolsky equation for weights $(1,2{\bf i}-4)$, applying prop. \ref{prop:TF2} we see that $\Phi^{\bf i}$ is generated by a representative, say $\Psi^{\bf i}$, of a cohomology class in $\check{H}^{1}(\PT,\mathcal{O}(2))$: 
\begin{align}\label{qKd}
\Phi^{\bf i}(x) = \frac{1}{2\pi\i}\oint_{\Gamma}\frac{\Psi^{\bf i}(x,\lambda)\lambda_A\d\lambda^A}{(o_B\lambda^B)^{4-{\bf i}}(\iota_B\lambda^B)^{\bf i}}.
\end{align}
Now, notice that, on restriction to any twistor line $\CP^1$, the group $\check{H}^{1}(\CP^1,\mathcal{O}(2))$ is trivial (recall \eqref{H1CP1}). 
This turns out to be manifested in the second item of the following proposition:

\begin{proposition}\label{prop:qkdeformations}
Let $(\M,g_{ab})$ be quaternionic-K\"ahler \eqref{qK}.
\begin{enumerate}[noitemsep, nolistsep]
\item Let $\Pi_{A'ABC}=\Pi_{A'(ABC)}$ be a spinor field satisfying 
\begin{align}
 \nabla^{A'}_{(A}\Pi_{BCD)A'} = 0, \label{EqForSpinorPotential}
\end{align}
and let $\gamma_{CDA'B'}:=\nabla_{(A'}^{A}\Pi_{B')ACD}$. Then the following identity holds:
\begin{align}
 \nabla_{(A}^{A'}\nabla_{B}^{B'}\gamma_{CD)A'B'} = 0. \label{zeroASDWeyl}
\end{align}
\item The linearised ASD Weyl spinor of the metric perturbation \eqref{metricpert} identically vanishes.
\end{enumerate}
\end{proposition}

\begin{proof}
For the first item, the calculation is tedious but straightforward: one has to use repeatedly the spinor form of the commutator $[\nabla_{AA'},\nabla_{BB'}]$ and related identities as in \cite[Section 4.9]{PR1}, together with the assumptions \eqref{qK} and \eqref{EqForSpinorPotential}. 

For the second item, we use \eqref{Wtensors} to write \eqref{metricpert} in spinor form:
\begin{align*}
 h^{\bf i}_{CDA'B'} &= \nabla_{(A'}^{A}\Pi^{\bf i}_{B')ACD}, 
 \qquad \text{where} \qquad 
 \Pi^{\bf i}_{B'ACD}:=(\nabla_{B'}^{B}+4f_{B'}^{B})S^{\bf i}_{ABCD}\Phi^{(i)}.
\end{align*}
Using that the weights of $S^{\bf i}_{ABCD}\Phi^{\bf i}$ are $w=3$, $p=0$, a short calculation gives $\Pi^{\bf i}_{B'ACD}=\C_{B'}^{B}[S^{\bf i}_{ABCD}\Phi^{\bf i}]$. We then get
\begin{align*}
\nabla_{(A}^{A'}\Pi^{\bf i}_{BCD)A'}=\C_{(A}^{A'}\Pi^{\bf i}_{BCD)A'}=S^{\bf i}_{(ABC}{}^{E}\C_{D)A'}\C^{A'}_{E}\Phi^{\bf i} = -\tfrac{1}{2}S^{\bf i}_{ABCD}\Box^{\C}\Phi^{\bf i},
\end{align*}
where in the first equality we used that $\Pi^{\bf i}_{B'ACD}$ has weights $w=2$, $p=0$, and in the last equality we used \eqref{CSD}. Since $\Phi^{\bf i}$ is a solution to $\Box^{\C}\Phi^{\bf i}=0$, then $\Pi^{\bf i}_{B'BCD}$ satisfies \eqref{EqForSpinorPotential}. Setting $\gamma_{CDA'B'}\equiv h^{\bf i}_{CDA'B'}$, we see that eq. \eqref{zeroASDWeyl} is satisfied for $h^{\bf i}_{ab}$. But the left hand side of \eqref{zeroASDWeyl} is precisely the formula for the linearised ASD Weyl spinor of a metric perturbation $\gamma_{ab}$, see \cite[Eq. (5.7.15)]{PR1}, thus the result follows.
\end{proof}

\subsection{Quaternionic deformations}

Since the metric perturbations \eqref{metricpert} have vanishing linearised Ricci and ASD Weyl, they preserve the conditions \eqref{qK} to linear order. We then say that the perturbations are {\em quaternionic deformations}. This fact allows us to connect the Teukolsky equations and the approach in the current work with the treatments of this kind of deformations in \cite{Hoegner} and \cite{Alexandrov2}. The latter are based on Przanowski's characterisation of any geometry satisfying \eqref{qK}: in \cite{Przanowski}, Przanowski showed that any solution to \eqref{qK} admits local complex coordinates $z^{\alpha}=(z^0,z^1)$ and a real scalar field $\K=\K(z^{\alpha},\bar{z}^{\alpha})$ such that the metric can be written as 
\begin{align}\label{metricprz}
 \d{s}^2 = \frac{6}{\lambda}\K_{\alpha\bar\beta}\d{z}^{\alpha}\d\bar{z}^{\beta} 
 + \frac{12}{\lambda}e^{\K}\d{z}^1\d\bar{z}^1
\end{align}
where $\K_{\alpha\bar\beta}=\partial_{\alpha}\partial_{\bar\beta}\K$. The self-dual Einstein condition \eqref{qK} is satisfied iff $\K$ solves Przanowski's equation ${\rm Prz}(\K)=0$, where
\begin{align}\label{prz}
 {\rm Prz}(\K):=\K_{0\bar{0}}\K_{1\bar{1}} - \K_{0\bar{1}}\K_{1\bar{0}} 
 +(2\K_{0\bar{0}} - \K_{0}\K_{\bar{0}})e^{\K}.
\end{align}

A class of quaternionic deformations is obtained by simply perturbing Przanowski's potential, $\K \to \K+\varepsilon\delta\K$, where $\delta\K$ is required to satisfy ${\rm dPrz}(\delta\K)=0$, with ${\rm dPrz}$ the linearised Przanowski operator:
\begin{equation}\label{lprz}
\begin{aligned}
 {\rm dPrz}(\delta\K) ={}& (\K_{1\bar{1}}+2e^{\K})\delta\K_{0\bar{0}}+\K_{0\bar{0}}\delta\K_{1\bar{1}} - \K_{1\bar{0}}\delta\K_{0\bar{1}} - \K_{0\bar{1}}\delta\K_{1\bar{0}} \\
 & + e^{\K}[(2\K_{0\bar{0}}-\K_{0}\K_{\bar{0}})\delta\K - \K_{0}\delta\K_{\bar{0}}-\K_{\bar{0}}\delta\K_{0}]
\end{aligned}
\end{equation}
where again $\delta\K_{\alpha\bar\beta}=\partial_{\alpha}\partial_{\bar\beta}\delta\K$.
In \cite{Hoegner}, H\"ogner constructed a conformally invariant differential operator on quaternionic-K\"ahler 4-manifolds, and showed that, when acting on scalar fields with conformal weight $w=1$, it is proportional to ${\rm dPrz}$. It turns out that H\"ogner's operator is essentially the quaternionic-K\"ahler specialisation of the Teukolsky operator of section \ref{sec:TeukolskyOp}. Using \eqref{Cscalar} and \eqref{Cvector}, one can indeed check that on scalars with weights $(w,0)$, 
\begin{align*}
 \Box^{\C} = \Box + (w+1)(2f^a\nabla_a+wf^af_a) - w\tfrac{R}{6}.
\end{align*}
In terms of the quantities in \eqref{metricprz}, the Lee form is $f_a\d{x}^a=-(\partial\log\K_{\bar{0}}+\bar\partial\log\K_{0})$. Then on any scalar with weights $w=1$, $p=0$, after a computation (and using $R=4\lambda$) we find
\begin{align}\label{BoxCPrz}
 \Box^{\C} = \frac{6}{\lambda\det g_{\alpha\bar\beta}}{\rm dPrz}
\end{align}
in agreement with the pioneering work \cite{Hoegner}. Furthermore, H\"ogner also gave a contour integral formula for perturbations of Przanowski's potential, see \cite[Eq. (5)]{Hoegner} (cf. also \cite{Alexandrov2}). This coincides with the case ${\bf i}=2$ in formula \eqref{qKd} (which corresponds to $w=1$, $p=0$).

While perturbations of Przanowski's potential give one class of quaternionic deformations, the results of section \ref{Sec:GPqK} show that there are other such deformations described by the tensors \eqref{metricpert} with ${\bf i}\neq2$. In these cases, the scalar equation involved is the Teukolsky equation with $p\neq0$ (instead of \eqref{lprz}), and \eqref{qKd} gives a contour integral formula for its solutions. Note, however, that in Euclidean signature the tensor fields \eqref{metricpert} with ${\bf i}\neq2$ are not real.

\subsection{The case with an isometry}

Finally, we comment very briefly on the case in which $(\M,g_{ab})$ satisfies \eqref{qK} and has also a Killing vector. In this case, Tod showed in \cite{Tod} that the geometry must be (locally) conformally K\"ahler, so the construction of section \ref{Sec:CK} applies. In the notation of that section, and using the expression for the Lee form given before, we deduce $\phi^{-1}=\K_{0}=\K_{\bar{0}}$. In particular, using \eqref{identityCK3} and \eqref{BoxCPrz}, for $w=1$ we recover \cite[Eq. (6.5)]{Alexandrov2}. More generally, we can map back and forth between the conformal wave equation and the Teukolsky equation for any weights $(w,p)$ with $p$ even, as described in section \ref{Sec:CK}.

\section{Final Remarks}

We constructed a correspondence between twistor functions and solutions to Teukolsky-like equations in conformally self-dual spaces, and we gave a contour integral formula for the latter in terms of essentially free holomorphic twistor data. Apart from the conformal self-duality assumption \eqref{SDWeyl}, the geometries we considered are generic (except in section \ref{Sec:QK} where we specialised to the self-dual Einstein case). To obtain more explicit coordinate descriptions of the space-time solutions, one must specialise to a specific geometry with its twistor space and use the corresponding incidence relations that express twistor coordinates in terms of space-time coordinates. For instance, in \cite[Section 5]{Hoegner} one can find examples of twistor coordinates for a number of geometries.

Following the twistor literature \cite{MW, DM}, there are other structures of interest that one can explore using the results in this work, such as the existence of hierarchies of commuting flows, Hamiltonian formulations, hidden symmetry algebras, the construction of special coordinates for twistor space $\PT$, the linear deformations of $\PT$ induced by gravitational perturbations, etc. A point of major interest for us would be to analyse whether the construction in this work can be used to study perturbations of {\em non-self-dual} black holes: in this case the background space also has an integrable complex structure and the connection $\C_a$ and many of its properties can be used (cf. \cite{A19}), but the twistor distribution has trivial kernel so ordinary twistor functions do not exist. We leave these questions for future work.

\paragraph{Acknowledgements.}
The author acknowledges support of the Institut Henri Poincar\'e (UAR 839 CNRS-Sorbonne Universit\'e) and LabEx CARMIN (ANR-10-LABX-59-01) in Paris, during a research stay in the spring 2024 in which this work was started.

\end{document}